%% file: main.tex
\documentclass[sigconf]{acmart}
\pdfoutput=1
\usepackage{url}
\usepackage{xspace,balance,multirow}
\usepackage[font=bf]{caption}
\usepackage[ruled, vlined, linesnumbered]{algorithm2e}
\usepackage{xcolor}
\usepackage{colortbl}
\usepackage{bbold}
\usepackage{enumitem}
\usepackage[captionskip=-1pt]{subfig}
\usepackage{wrapfig}
\usepackage{flushend}
\usepackage{tikz}
\usepackage{pgfplots}
\usepackage{pgf}
\usetikzlibrary{patterns}
\pgfplotsset{compat=1.16}
\usepackage{graphicx}

\AtBeginDocument{%
	\providecommand\BibTeX{{%
			\normalfont B\kern-0.5em{\scshape i\kern-0.25em b}\kern-0.8em\TeX}}}

\copyrightyear{2023} 
\acmYear{2023} 
\setcopyright{rightsretained} 
\acmConference[WWW '23]{Proceedings of the ACM Web Conference 2023}{April 30-May 4, 2023}{Austin, TX, USA}
\acmBooktitle{Proceedings of the ACM Web Conference 2023 (WWW '23), April 30-May 4, 2023, Austin, TX, USA}
\acmDOI{10.1145/3543507.3583408}
\acmISBN{978-1-4503-9416-1/23/04}

\input{cmd}

\begin{document}
	
	\title{Node-wise Diffusion for Scalable Graph Learning}
	
	\author{Keke Huang}
	\email{kkhuang@nus.edu.sg}
	\affiliation{%
		\institution{National University of Singapore}
		\country{}
	}
	\author{Jing Tang}
	\email{jingtang@ust.hk}
	\affiliation{%
		\institution{The Hong Kong University of Science and Technology (Guangzhou)}
		\institution{The Hong Kong Uni. of Sci. and Tech.}
		\country{}
	}
	
	\author{Juncheng Liu}
	\email{juncheng.liu@u.nus.edu}
	\affiliation{%
		\institution{National University of Singapore}
		\country{}
	}
	
	\author{Renchi Yang}
	\email{renchi@hkbu.edu.hk}
	\affiliation{%
		\institution{Hong Kong Baptist University}
		\country{}
	}
	
	\author{Xiaokui Xiao}
	\email{xkxiao@nus.edu.sg}
	\affiliation{%
		\institution{National University of Singapore}
		\institution{CNRS@CREATE, Singapore}
		\country{}
	}
	
	\begin{abstract}
		Graph Neural Networks (GNNs) have shown superior performance for semi-supervised learning of numerous web applications, such as classification on web services and pages, analysis of online social networks, and recommendation in e-commerce. The state of the art derives representations for {\em all} nodes in graphs following the {\em same} diffusion (message passing) model without discriminating their uniqueness.
		However, (i) labeled nodes involved in model training usually account for a small portion of graphs in the semi-supervised setting, and (ii) different nodes locate at different graph local contexts and it inevitably degrades the representation qualities if treating them undistinguishedly in diffusion. 
		
		To address the above issues, we develop \NDM, a universal node-wise diffusion model, to capture the unique characteristics of each node in diffusion, by which \NDM is able to yield high-quality node representations. In what follows, we customize \NDM for semi-supervised learning and design the \gtro model. In particular, \gtro advances the efficiency significantly since it (i) produces representations for labeled nodes only and (ii) adopts well-designed neighbor sampling techniques tailored for node representation generation. Extensive experimental results on various types of web datasets, including citation, social and co-purchasing graphs, not only verify the state-of-the-art effectiveness of \gtro but also strongly support the remarkable scalability of \gtro. In particular, \gtro completes representation generation and training within $10$ seconds on the dataset with hundreds of millions of nodes and billions of edges, up to orders of magnitude speedups over the baselines, while achieving the highest F1-scores on classification \footnote{The code of \gtro can be accessed at \url{https://github.com/kkhuang81/NIGCN}.}.
	\end{abstract}



	\maketitle
	\begin{sloppy}
		\input{content}

		\bibliographystyle{ACM-Reference-Format}
		\bibliography{ref}
		
		\clearpage
		\appendix
		\input{appendix}
	\end{sloppy}
	
\end{document}

%% file: cmd.tex
\newcommand{\ie}{{i.e.,}\xspace}
\newcommand{\eg}{{e.g.,}\xspace}
\newcommand{\spara}[1]{\vspace{1mm}\noindent\textbf{#1.}}
\newcommand{\savespace}[1]{\ignorespaces}
\SetKw{KwAnd}{and}
\SetKw{KwOr}{or}
\SetKw{KwBreak}{break}
\SetKw{KwContinue}{continue}

\newcommand{\gtro}{\textsf{NIGCN}\xspace}

\newcommand{\GCN}{\textsf{GCN}\xspace}
\newcommand{\SGC}{\textsf{SGC}\xspace}
\newcommand{\GAT}{\textsf{GAT}\xspace}
\newcommand{\GDC}{\textsf{GDC}\xspace}
\newcommand{\APPNP}{\textsf{APPNP}\xspace}
\newcommand{\PPRGo}{\textsf{PPRGo}\xspace}
\newcommand{\GraphSAGE}{\textsf{GraphSAGE}\xspace}
\newcommand{\GraphSAINT}{\textsf{GraphSAINT}\xspace}
\newcommand{\FastGCN}{\textsf{FastGCN}\xspace}
\newcommand{\GBP}{\textsf{GBP}\xspace}
\newcommand{\AGP}{\textsf{AGP}\xspace}
\newcommand{\NDLS}{\textsf{NDLS}\xspace}
\newcommand{\ShaDowGCN}{\textsf{ShaDow-GCN}\xspace}
\newcommand{\Grandp}{\textsf{GRAND\({+}\)}\xspace}

\newcommand{\UDF}{\textsf{UDF}\xspace}
\newcommand{\NDM}{\textsf{NDM}\xspace}
\newcommand{\GHD}{\textsf{GHD}\xspace}

\newcommand{\A}{\mathbf{A}\xspace}
\newcommand{\D}{\mathbf{D}\xspace}

\newcommand{\G}{\mathbf{G}\xspace}
\newcommand{\I}{\mathbf{I}\xspace}
\renewcommand{\P}{\mathbf{P}\xspace}
\newcommand{\U}{\mathbf{U}\xspace}

\newcommand{\X}{\mathbf{X}\xspace}
\newcommand{\Z}{\mathbf{Z}\xspace}

\renewcommand{\H}{\mathbf{H}\xspace}
\renewcommand{\L}{\mathbf{L}\xspace}

\newcommand{\TA}{\tilde{\A}\xspace}

\newcommand{\1}{\mathbf{1}\xspace}

\newcommand{\x}{\mathbf{x}\xspace}
\newcommand{\z}{\mathbf{z}\xspace}

\newcommand{\T}{\mathcal{T}\xspace}
\newcommand{\R}{\mathbb{R}\xspace}
\newcommand{\N}{\mathcal{N}\xspace}
\renewcommand{\S}{\mathcal{S}\xspace}

\newcommand{\V}{\mathcal{V}\xspace}
\newcommand{\E}{\mathcal{E}\xspace}
\renewcommand{\d}{{\ensuremath{\mathrm{d}}}}
\newcommand{\e}{{\ensuremath{\mathrm{e}}}}

\newtheorem{property}{Property}[section]

\newenvironment{customlegend}[1][]{%
    \begingroup
    \csname pgfplots@init@cleared@structures\endcsname
    \pgfplotsset{#1}%
}{%
    \csname pgfplots@createlegend\endcsname
    \endgroup
}%

\def\addlegendimage{\csname pgfplots@addlegendimage\endcsname}

\pgfplotsset{every tick label/.append style={font=\tiny}}
\definecolor{myblue}{RGB}{113,210,242}
\definecolor{myorange}{RGB}{249,205,173}
\definecolor{myred}{RGB}{183,0,162}
\definecolor{mygreen}{RGB}{0,191,183}
\definecolor{myredd}{RGB}{190,0,66}

%% file: content.tex
\section{Introduction}\label{sec:intro}

In recent years, Graph Neural Networks (GNNs) have gained increasing attention in both academia and industry due to their superior performance on numerous web applications, such as classification on web services and pages~\cite{zhang2021bilinear, GuoC21}, image search~\cite{ImageSearch}, web spam detection~\cite{BelahcenBS15}, e-commerce recommendations~\cite{YingHCEHL18, Fan0LHZTY19, WuZGHWGC19}, and social analysis~\cite{LiG19, QiuTMDW018, SankarLYS21}. Various GNN models have been developed~\cite{VelickovicCCRLB18,XuLTSKJ18,ChenMX18,ZouHWJSG19,KlicperaBG19,KlicperaWG19,BojchevskiKPKBR20,ChienP0M21,ZhangYSLOTYC21,FengDHYCK022} accordingly. Among them, {\em semi-supervised classification} is one of the most extensively studied problems due to the scarce labeled data in real-world applications~\cite{KaiserNRB17, VartakTMBL17, DingWLSLL20}.

Graph Convolutional Network (\GCN)~\cite{KipfW17} is the seminal GNN model proposed for semi-supervised classification. \GCN conducts {\it feature propagation} and {\it transformation} recursively on graphs and is trained in a full-batch manner, thus suffering from severe scalability issues~\cite{HamiltonYL17, ChenZS18, ChenMX18, WuSZFYW19, ZouHWJSG19, ZengZSKP20, WangHW000W21}. Since then, there has been a large body of research on improving the efficiency.
One line of work focuses on utilizing sampling and preprocessing techniques. Specifically, \GraphSAGE~\cite{HamiltonYL17} and \FastGCN~\cite{ChenMX18} sample a fixed number of neighbors for each layer. \GraphSAINT~\cite{ZengZSKP20} and \ShaDowGCN~\cite{ZengZXSMKPJC21} randomly extract subgraphs with limited sizes as training graphs. \textsf{Cluster-GCN}~\cite{ChiangLSLBH19} partitions graphs into different clusters and then randomly chooses a certain number of clusters as training graphs. Another line of research decouples feature propagation and transformation to ease feature aggregations. In particular, \SGC~\cite{WuSZFYW19} proposes to remove non-linearity in transformation and multiplies the feature matrix to the $K$-th power of the normalized adjacency matrix for feature aggregation. Subsequently, a plethora of decoupled models are developed to optimize the efficiency of feature aggregation by leveraging various graph techniques, including \APPNP~\cite{KlicperaBG19}, \GBP~\cite{ChenWDL00W20}, \AGP~\cite{WangHW000W21}, and \Grandp~\cite{FengDHYCK022}. 

Despite the efficiency advances, current models either calculate node presentations for enormous unlabeled nodes or ignore the unique topological structure of each labeled node during representation generation. Therefore, there is still room for improvement in efficiency and effectiveness. To explain, labeled nodes involved in model training in semi-supervised learning usually take up a small portion of graphs, especially on massive graphs, and computing representations for all nodes in graphs is unnecessarily inefficient. Meanwhile, different nodes reside in different graph locations with distinctive neighborhood contexts. Generating node representations without considering their topological uniqueness inevitably degrades the representation qualities. 

To remedy the above deficiencies, we first develop a node-wise diffusion model \NDM. Specifically, \NDM calculates an individual diffusion length for each node by taking advantage of the unique topological characteristic for high-quality node representations. In the meantime, \NDM employs a  universal diffusion function \GHD adaptive to various graphs. In particular, \GHD is a {\em general heat diffusion} function that is capable of capturing different diffusion patterns on graphs with various densities. By taking \NDM as the diffusion model for feature propagations, we design \gtro (\underline{\bf N}ode-w\underline{\bf I}se GCN), a GCN model with superb scalability. In particular, \gtro only computes representations for the labeled nodes for model training without calculating (hidden) representations for any other nodes. In addition, \gtro adopts customized neighbor sampling techniques during diffusion. By eliminating those unimportant neighbors with noise features, our neighbor sampling techniques not only improve the performance of \gtro for semi-supervised classification but also boost the efficiency significantly.  

We evaluate \gtro on $7$ real-world datasets and compare with $13$ baselines for transductive learning and $7$ competitors for inductive learning. Experimental results not only verify the superior performance of \gtro for semi-supervised classification but also prove the remarkable scalability of \gtro. In particular, \gtro completes feature aggregations and training within $10$ seconds on the dataset with hundreds of millions of nodes and billions of edges, up to orders of magnitude speedups over the baselines, while achieving the highest F1-scores on classification.

In a nutshell, our contributions are summarized as follows.
\begin{itemize}[topsep=1mm,partopsep=0pt,itemsep=1mm,leftmargin=18pt]
\item We propose a node-wise diffusion model \NDM. \NDM customizes each node with a unique diffusion scheme by utilizing the topological characteristics and provides a general heat diffusion function capable of capturing different diffusion patterns on graphs with various densities.
\item We design a scalable GNN model \gtro upon \NDM. \gtro calculates node representation for a small portion of labeled nodes without producing intermediate (hidden) representations for any other nodes. Meanwhile, neighbor sampling techniques adopted by \gtro further boost its scalability significantly. 
\item We conduct comprehensive experiments to verify the state-of-the-art performance of \gtro for semi-supervised classification and the remarkable scalability of \gtro. 

\end{itemize}

\section{Related Work}\label{sec:relatedwork}

\citet{KipfW17} propose the seminal Graph Convolutional Network (\GCN) for semi-supervised classification. However, \GCN suffers from severe scalability issues since it executes the feature propagation and transformation recursively and is trained in a full-batch manner. To alleviate the pain, two directions, \ie decoupled models and sampling-based models, have been explored.

\spara{Decoupled Models} \SGC proposed by \citet{WuSZFYW19} adopts the decoupling scheme by removing non-linearity in feature transformation and propagates features of neighbors within $K$ hops directly, where $K$ is an input parameter. Following \SGC, a plethora of decoupled models have been developed. To consider node proximity, \APPNP~\cite{KlicperaBG19} utilizes personalized PageRank (PPR)~\cite{page1999pagerank, ScarselliTH04} as the diffusion model and takes PPR values of neighbors as aggregation weights. To improve the scalability, \PPRGo~\cite{BojchevskiKPKBR20} reduces the number of neighbors in aggregation by selecting neighbors with top-$K$ PPR values after sorting them. Graph diffusion convolution (\GDC)~\cite{KlicperaWG19} considers various diffusion models, including both PPR and heat kernel PageRank (HKPR) to capture diverse node relationships. Later, \citet{ChenWDL00W20} apply generalized PageRank model~\cite{0005CM19} and propose \GBP that combines reverse push and random walk techniques to approximate feature propagation. \citet{WangHW000W21} point out that \GBP consumes a large amount of memory to store intermediate random walk matrices and propose \AGP that devises a unified graph propagation model and employs forward push and random sampling to select subsets of unimportant neighborhoods so as to accelerate feature propagation. \citet{ZhangYSLOTYC21} consider the number of neighbor hops before the aggregated feature gets smoothing. To this end, they design \NDLS  and calculate an individual local-smoothing iteration for each node on feature aggregation. Recently, \citet{FengDHYCK022} investigate the graph random neural network (GRAND) model. To improve the scalability, they devise \Grandp by leveraging a generalized forward push to compute the propagation matrix for feature aggregation. In addition, \Grandp only incorporates neighbors with top-K values for further scalability improvement.

\spara{Sampling-based Models} To avoid the recursive neighborhood over expansion, \GraphSAGE~\cite{HamiltonYL17} simply samples a fixed number of neighbors uniformly for each layer. Instead of uniform sampling, \FastGCN~\cite{ChenMX18} proposes importance sampling on neighbor selections to reduce sampling variance. Subsequently, \textsf{AS-GCN}~\cite{Huang0RH18} considers the correlations of sampled neighbors from upper layers and develops an adaptive layer-wise sampling method for explicit variance reduction. To guarantee the algorithm convergence, \textsf{VR-GCN} proposed by \citet{ChenZS18} exploits historical hidden representations as control variates and then reduces sampling variance via the control variate technique. Similar to \textsf{AS-GCN}, \textsf{LADIES}~\cite{ZouHWJSG19} also takes into account the layer constraint and devises a layer-wise, neighbor-dependent, and importance sampling manner, where two graph sampling methods are proposed as a consequence. \textsf{Cluster-GCN}~\cite{ChiangLSLBH19} first applies graph cluster algorithms to partition graphs into multiple clusters, and then randomly takes several clusters as training graphs. Similarly, \GraphSAINT~\cite{ZengZSKP20} samples subgraphs as new training graphs, aiming to improve the training efficiency. \citet{HuangZXLZ21} adopt the graph coarsening method developed by~\citet{Loukas19} to reshape the original graph into a smaller graph, aiming to boost the scalability of graph machine learning. Lately, \citet{ZengZXSMKPJC21} propose to extract localized subgraphs with bounded scopes and then run a GNN of arbitrary depth on it. This principle of decoupling GNN scope and depth, named as \textsf{ShaDow}, can be applied to existing GNN models.

However, all the aforementioned methods either (i) generate node representations for {\em all} nodes in the graphs even though labeled nodes in training are scarce or (ii) overlook the topological uniqueness of each node during feature propagation. Ergo, there is still room for improvement in both efficiency and efficacy.

\section{Node-Wise Diffusion Model}\label{sec:modeldesign}

In this section, we reveal the weakness in existing diffusion models and then design \NDM, consisting of two core components, \ie (i) the diffusion matrix and the diffusion length for each node, and (ii) the universal diffusion function generalized to various graphs. 

\subsection{Notations}\label{sec:notation}
\vspace{-1mm}
For the convenience of expression, we first define the frequently used notations. We use calligraphic fonts, bold uppercase letters, and bold lowercase letters to represent sets (\eg $\N$), matrices (\eg $\A$), and vectors (\eg $\x$), respectively. The $i$-th row (resp.\ column) of matrix $\A$ is represented by $\A[i,\cdot]$ (resp.\ $\A[\cdot, i]$). 

Let $\G=(\V, \E, \X)$ be an undirected graph where $\V$ is the node set with $|\V|=n$, $\E$ is the edge set with $|\E|=m$, and $\X \in \R^{n\times f}$ is the feature matrix. Each node $v\in \V$ is associated with a $f$-dimensional feature vector $\x_v \in \X$. For ease of exposition, node $u \in \V$ also indicates its index. Let $\N_u$ be the {\em direct} neighbor set and $d_u=|\N_u|$ be the degree of node $u$. Let $\A \in \R^{n\times n}$ be the adjacency matrix of $\G$, \ie $\A[u,v]=1$ if $\langle u, v \rangle \in \E$; otherwise $\A[u,v]=0$, and $\D \in \R^{n\times n}$ be the diagonal degree matrix of $\G$, \ie $\D[u,u]=d_u$. Following the convention~\cite{ChenWDL00W20,WangHW000W21}, we assume that $\G$ is a {\em self-looped} and connected graph.

\subsection{Diffusion Matrix and Length}\label{sec:diffmatrixlen}

\spara{Diffusion Matrix} Numerous variants of Laplacian matrix are widely adopted as diffusion matrix in existing GNN models~\cite{KipfW17, WuSZFYW19, KlicperaBG19, ChenWDL00W20, Liu0LH21, ZhangYSLOTYC21}. Among them, the transition matrix $\P=\D^{-1}\A$ is intuitive and easy-explained. Let $1= \lambda_1 \ge\lambda_2  \ge \ldots \ge \lambda_n > -1$ be the eigenvalues of $\P$. During an infinite diffusion, any initial state $\pi_0 \in \R^n$ of node set $\V$ converges to the stable state $\pi$, \ie $\pi=\lim_{\ell \to \infty}\pi_0\P^\ell $ where $\pi(v)=\frac{d_v}{2m}$.

\spara{Diffusion Length} As stated, different nodes reside at different local contexts in the graphs, and the corresponding receptive fields for information aggregation differ. Therefore, it is rational that each node $u$ owns a unique length $\ell_u$ of diffusion steps. As desired, node $u$ aggregates informative signals from neighbors within the range of $\ell_u$ hops while obtaining limited marginal information out of the range due to over-smoothing issues. To better quantify the effective vicinity, we first define $\tau$-distance as follows. 
\begin{definition}[$\tau$-Distance]\label{def:eps-vicinity}
Given a positive constant $\tau$ and a graph $\G=(\V, \E)$ with diffusion matrix $\P$, a length $\ell$ is called $\tau$-distance of node $u\in \V$ if it satisfies that for every $v \in \V$, $\frac{|\P^\ell[u,v]-\pi(v)|}{\pi(v)} \le \tau$.
\end{definition}
According to Definition~\ref{def:eps-vicinity}, $\ell_u$ being $\tau$-distance of $u$ ensures that informative signals from neighbors are aggregated. On the other hand, to avoid over-smoothing, $\ell_u$ should not be too large. In the following, we provide an appropriate setting of $\ell_u$ fitting both criteria.

\begin{theorem}\label{thm:diffusionstep}
Given a positive constant $\tau$ and a graph $\G=(\V, \E)$ with diffusion matrix $\P$, $\ell_u :=\left\lceil\log_\lambda{\frac{\tau \sqrt{d_{\min}d_u}}{2m}}\right\rceil$ is $\tau$-distance of node $u$, where $\lambda=\max\{\lambda_2,-\lambda_n\}$ and $d_{\min}=\min\{d_v\colon v \in \V\}$.  
\end{theorem}
\begin{proof}[Proof of Theorem~\ref{thm:diffusionstep}]
Let $\mathbf{e}_u\in \R^{1\times n}$ be a one-hot vector having $1$ in coordinate $u \in \V$ and $\1_n \in \R^{1\times n}$ be the $1$-vector of size $n$. Then, $\P^\ell [u,v] =\mathbf{e}_u\P^\ell\mathbf{e}^\top_v$. Let $\tilde{\P}=\D^{1/2}\P\D^{-1/2}=\D^{-1/2}\A\D^{-1/2}$ and $\mathbf{u}^\top_i$ be the corresponding eigenvector of its $i$-th eigenvalue (sorted in descending order) of $\tilde{\P}$. For $\mathbf{e}_u$ and $\mathbf{e}_v$, we decompose
\begin{equation*}
    \mathbf{e}_u\D^{-1/2}=\sum\nolimits^n_{i=1}\alpha_i \mathbf{u}_i, \text{ and } \mathbf{e}_v\D^{1/2}=\sum\nolimits^n_{i=1}\beta_i \mathbf{u}_i.
\end{equation*}
Note that $\{\mathbf{u}^\top_1,\dotsc,\mathbf{u}^\top_n\}$ form the orthonormal basis and $\mathbf{u}_1=\frac{\1_n \D^{1/2}}{\sqrt{2m}}$. Thus, we have 
$\alpha_1=\mathbf{e}_u\D^{-1/2}\mathbf{u}^\top_1 = \frac{1}{\sqrt{2m}}$ and $\beta_1=\mathbf{e}_v\D^{1/2}\mathbf{u}^\top_1= \frac{d_v}{\sqrt{2m}}$. Since  $\tilde{\P}$ is the similar matrix of $\P$, they share the same eigenvalues. Therefore, we have
\begin{align*}
&\frac{\left|\P^\ell [u,v] -\pi(v) \right|}{\pi(v)} 
=\frac{\left|\mathbf{e}_u\P^\ell\mathbf{e}^\top_v - \pi(v) \right|}{\pi(v)} 
=\frac{\left|\mathbf{e}_u\D^{-1/2}\tilde{\P}^\ell\D^{1/2}\mathbf{e}^\top_v - \pi(v) \right|}{\pi(v)} \\
&=\frac{ \left|\sum^{n}_{i=1}\alpha_i\beta_i\lambda_i^\ell-\pi(v) \right|}{\pi(v)} 
=\frac{ \left|\sum^{n}_{i=2}\alpha_i\beta_i\lambda_i^\ell \right|}{\pi(v)} 
\le \lambda^\ell \cdot \frac{ \sum^{n}_{i=2} \left|\alpha_i\beta_i \right|}{\pi(v)} \\
&\le \lambda^\ell\cdot  \frac{ \| \mathbf{e}_u \D^{-1/2}\| \| \mathbf{e}_v \D^{1/2}\|}{ d_v/2m} 
= \frac{2m\lambda^\ell}{\sqrt{ d_vd_u}},
\end{align*}
where the second inequality is by Cauchy–Schwarz inequality. Finally, setting $ \ell:=\left\lceil\log_\lambda{\frac{\tau \sqrt{d_{\min}d_u}}{2m}}\right\rceil$ completes the proof.
\end{proof}
For the $\ell_u$ defined in Theorem~\ref{thm:diffusionstep}, it is $\tau$-distance of node $u$ and in the meantime involves the topological uniqueness of node $u$. Moreover, the performance can be further improved by tuning the hyperparameter $\tau$.

\subsection{Universal Diffusion Function}\label{sec:diffunc}

As we know, the diffusion model defined by the symmetrically normalized Laplacian matrix $\L=\I-\D^{-1/2}\A\D^{-1/2}$ is derived from {\em Graph Heat Equation}~\cite{chung1997spectral,wang2021dissecting}, \ie
\begin{equation}\label{eqn:heateqn}
\frac{\d \H_t}{\d t}=-\L\H_t, \text{ and } \H_0=\X,
\end{equation}
where $\H_t$ is the node status of graph $\G$ at time $t$. By solving the above differential function, we have 
\begin{equation}
\H_t=\e^{-t\L}\X=\e^{-t(\I-\TA)}\X=\e^{-t}\sum^\infty_{\ell=0}\frac{t^\ell\TA^\ell}{\ell\,!}\X,   
\end{equation}
where $\TA=\D^{-1/2}\A\D^{-1/2}$. In this regard, the underlying diffusion follows the {\em Heat Kernel PageRank} (HKPR) function as
\begin{equation}\label{eqn:heatkernel}
f(\omega, \ell)=\e^{-\omega}\frac{\omega^\ell}{\ell\,!},    
\end{equation}
where $\omega \in \mathbb{Z}^+$ is the parameter. However, $f(\omega, \ell)$ is neither expressive nor general enough to act as the universal diffusion function for real-world graphs, hinted by the following graph property.

\begin{property}[\cite{chung1997spectral}]\label{pro:prop1}
For graph $\G$ with average degree $d_\G$, we have $1-\Delta_\lambda=O(\frac{1}{\sqrt{d_\G}})$ where $\Delta_\lambda$ is the spectral gap of $\G$.
\end{property}

For the diffusion matrix $\P$ defined on $\G$, we have $\lambda=1-\Delta_\lambda$. Meanwhile, according to the analysis of Theorem~\ref{thm:diffusionstep}, we know that $\P^\ell [u,v] -\pi(v)=\sum^{n}_{i=2}\alpha_i\beta_i\lambda_i^\ell$, representing the convergence, is (usually) dominated by $\lambda^\ell$. As a result, diffusion on graphs with different densities, \ie~$d_\G$, converges at different paces. In particular, sparse graphs with small $d_\G$ incurring large $\lambda$ tend to incorporate neighbors in a long range while dense graphs with large $d_\G$ incurring small $\lambda$ are prone to aggregate neighbors not far away. In addition, it has been widely reported in the literature~\cite{KlicperaWG19, FengDHYCK022} that different graphs ask for different diffusion functions, which is also verified by our experiments in Section~\ref{sec:performance}.

To serve the universal purpose, a qualified diffusion function should be able to (i) expand smoothly in long ranges, (ii) decrease sharply in short intervals, and (iii) peak at specified hops, as required by various graphs accordingly. Clearly, the HKPR function in \eqref{eqn:heatkernel} fulfills the latter two requirements but fails the first one since it decreases exponentially when $\ell \ge \omega$. One may propose to consider {\em Personalized PageRank} (PPR). However, the PPR function is monotonically decreasing and thus cannot reach condition (iii). 

\begin{figure}[!bpt]
\centering
\begin{small}
\hspace{-1mm}
\subfloat[Smooth]{
\begin{tikzpicture}[scale=1,every mark/.append style={mark size=1.5pt}]
    \begin{axis}[
        height=\columnwidth/2.5,
        width=\columnwidth/2.4,
        ylabel={\em $\omega=1.1, \ \rho=0.05$},
        xlabel={\em $\ell$},
        xlabel style={yshift=0.1cm},
        xmin=0, xmax=20.5,
        ymin=0, ymax=0.1,
        xtick={0,5,10,15,20},
        ytick={0,0.1},
        yticklabel style = {font=\scriptsize},
        ylabel style={font=\footnotesize},
        every axis y label/.style={font=\scriptsize,at={(current axis.north west)},right=8mm,above=0mm},
    ] 
    \addplot[line width=0.4mm,color=blue] file[skip first]{smooth.dat}; 
    \end{axis}
\end{tikzpicture}}
\subfloat[PPR]{
\begin{tikzpicture}[scale=1,every mark/.append style={mark size=1.5pt}]
    \begin{axis}[
        height=\columnwidth/2.5,
        width=\columnwidth/2.4,
        ylabel={\em $\omega=0.5, \ \rho=0$},
        xlabel={\em $\ell$},
        xlabel style={yshift=0.1cm},
        xmin=0, xmax=10.5,
        ymin=0, ymax=0.6,
        xtick={0,1,2,3,4,5,6,7,8,9,10},
        ytick={0,0.2,0.4,0.6},
        yticklabel style = {font=\scriptsize},
        ylabel style={font=\footnotesize},
        every axis y label/.style={font=\scriptsize,at={(current axis.north west)},right=8mm,above=0mm},
    ] 
    \addplot[line width=0.4mm,color=blue] file[skip first]{PPR.dat}; 
    \end{axis}
\end{tikzpicture}\hspace{0mm}}
\subfloat[HKPR]{
\begin{tikzpicture}[scale=1,every mark/.append style={mark size=1.5pt}]
    \begin{axis}[
        height=\columnwidth/2.5,
        width=\columnwidth/2.4,
        ylabel={\em $\omega=5.0, \ \rho=1.0$},
        xlabel={\em $\ell$},
        xlabel style={yshift=0.1cm},
        xmin=0, xmax=15.5,
        ymin=0, ymax=0.2,
        xtick={0,3,6,9,12,15},
        ytick={0,0.1,0.2},
        yticklabel style = {font=\scriptsize},
        ylabel style={font=\footnotesize},
        every axis y label/.style={font=\scriptsize,at={(current axis.north west)},right=8mm,above=0mm},
    ] 
    \addplot[line width=0.4mm,color=blue] file[skip first]{HKPR.dat}; 
    \end{axis}
\end{tikzpicture}\hspace{0mm}}
\end{small}
\caption{Three exemplary expansion tendency of \GHD.} \label{fig:udf}
\end{figure}
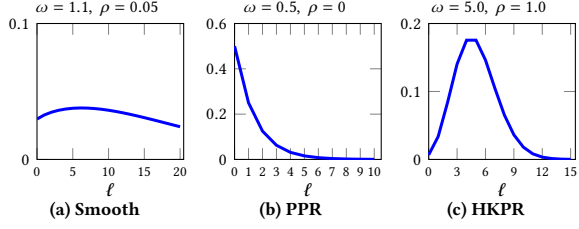

Inspired by the above analysis, we try to ameliorate $f(\omega, \ell)$ to a universal diffusion function with a controllable change tendency for general purposes. To this end, we extend the graph heat diffusion Equation~\eqref{eqn:heatkernel} by introducing an extra {\em power parameter} $\rho \in \R^{+}$ and devise our {\em General Heat Diffusion} (\GHD) function as 
\begin{equation}\label{eqn:unidiff}
U(\omega, \rho, \ell) =\frac{\omega^\ell}{ (\ell\,!)^\rho\cdot C}
\end{equation}
for the diffusion weight at the $\ell$-th hop, where $\omega \in \R^{+}$ is the new heat parameter and $C= \sum^{\infty}_{\ell=0} \frac{\omega^\ell}{(\ell\,!)^\rho}$ is the normalization factor.

As desired, \GHD can be regarded as a general extension of the graph heat diffusion model, and parameters $\omega$ and $\rho$ together determine the expansion tendency. In particular, it is trivial to verify that \GHD is degraded into HKPR when $\rho=1$, and \GHD becomes PPR when $\rho=0$. As illustrated in \figurename~\ref{fig:udf}, by setting different $\omega$ and $\rho$ combinations, \GHD is able to exhibit smooth, exponential (\ie~PPR), or peak expansion (\ie~HKPR) tendency.

\subsection{Diffusion Model Design}\label{sec:model}

Upon $\tau$-distance and diffusion function \UDF, our node-wise diffusion model (\NDM) can be concreted. Specifically, given a target set $\T \subseteq \V$, the representation $\Z_\T$ under \NDM is calculated as
\begin{equation}\label{eqn:ndm}
    \Z_\T=\sum^{L}_{\ell=0}\mathbf{U}\Gamma \P^\ell\X,
\end{equation}
where $L=\max\{\ell_u\colon \forall u\in \T\}$, $\mathbf{U}=\mathrm{Diag}\{U(\omega, \rho, \ell)\colon \forall u \in \T\} \in \R^{|\T|\times |\T|}$ is a diagonal matrix, and $\Gamma=\I[\T, \cdot] \in \R^{|\T| \times n}$ is the indicator matrix, \ie $\Gamma[i,u]=1$ if $\T[i]=u$ and $\Gamma[i,u]=0$ otherwise. 

\begin{algorithm}[!t]
\setlength{\hsize}{0.95\linewidth}
	\caption{Node-wise Diffusion Model}
	\label{alg:NDM}
	\KwIn{Graph $\G$, feature matrix $\X$, target set $\T$,\newline and hyperparameters $\tau$, $\omega$, $\rho$}
	\KwOut{Representation $\Z_\T$}
    $d_{\max}\gets \max_{u\in \T}\{d_u\}$\;
    $L\gets \left\lceil\log_\lambda{\frac{\tau \sqrt{d_{\min}d_{\max}}}{2m}}\right\rceil$\;
	$U$ is calculated according to~\eqref{eqn:unidiff}\;
	$\Gamma\gets \I[\T, \cdot]$, $\Z_\T \gets \mathbf{0}^{|\T|\times f}$\;
	\For{$\ell\gets 0$ \KwTo $L$}
	{
	    $\mathbf{U} \gets \textrm{Diag}\{U(\omega, \rho, \ell)\colon \forall u \in \T\}$\;
	    $\Z_\T \gets \Z_\T + \mathbf{U}\Gamma$\; 
	    $\Gamma \gets \Gamma \P$, $\ell \gets \ell+1$\;  
	} 
	$\Z_\T\gets \Z_\T\X$\; 
	\Return $\Z_\T$\;
\end{algorithm}

The pseudo-code of \NDM is presented in Algorithm~\ref{alg:NDM}. \NDM first finds the largest degree $d_{\max}$ for nodes $\T$, and computes the corresponding $\tau$-distance as $L$. Then, \NDM accumulates the weights of neighbors within $L$ ranges for each node $u\in \T$, recorded as $\Z_\T$. Note that $\mathbf{U}[u,u]=0$ if $\ell > L$. Finally, representation $\Z_\T$ is calculated by multiplying the feature matrix $\X$.

\spara{Time Complexity} It takes $O(m)$ time to calculate $\lambda$ using the iterative methods~\cite{James1997}, and hence computing $L$ take $O(m+|\T|)$ time. Matrix multiplications $\mathbf{U}\Gamma$ and $\Gamma \P$ dominate the running time, which takes time complexities of $O(n|\T|)$ and $O(m|\T|)$, respectively. Therefore, as it takes $O(f|\T|)$ time to compute $Z_\T\X$, the total time complexity of \NDM is $O((m+n)L|\T|+f|\T|)$. 

\section{Optimization in Node Representation Learning}\label{sec:optimization} 

Algorithm~\ref{alg:NDM} in Section~\ref{sec:modeldesign} presents a general node-wise diffusion model. However, it is yet optimal to be applied to reality. In this section, we aim to instantiate \NDM in a practical manner and optimize the procedure of feature propagations.
\subsection{Instantiation of \NDM}\label{sec:instant}

\spara{Practical Implementation of $\tau$-Distance} Calculating the $\tau$-distance of each node is one of the critical steps in \NDM, which requires the second largest eigenvalue $\lambda$ of the diffusion matrix. However, it is computationally expensive to compute $\lambda$ for large graphs. To circumvent the scenario, we employ property~\ref{pro:prop1} to substitute $\lambda$ without damaging the efficacy of \NDM.

As we analyze in Section~\ref{sec:diffunc}, according to Property~\ref{pro:prop1}, we borrow a correction factor $C_\G$ specific for graph $\G$ to ensure $\lambda=1-\Delta_\lambda=\frac{C_\G}{\sqrt{d_\G}}$. Meanwhile, for the sake of practicality, we could merge hyperparameter $\tau$ and $C_\G$ into one tunable parameter $\tau^\prime$ to control the bound of $\tau$-distance $\ell_u$ such that 
\begin{equation}\label{eqn:newrange}
\ell_u=\log_\lambda{\frac{\tau \sqrt{d_{\min}d_u}}{2m}}=\frac{\ln{\frac{2m}{\sqrt{d_{\min}d_u}}}-\ln\tau}{\ln{\sqrt{d_\G}}-\ln C_\G}:= \tau^\prime\frac{\ln{\frac{2m}{\sqrt{d_{\min}d_u}}}}{\ln{\sqrt{d_\G}}}.
\end{equation}

\spara{Important Neighbor Identification and Selection} \NDM in Algorithm~\ref{alg:NDM} aggregates all neighbors during diffusion for each node, which, however, is neither effective nor efficient. The rationale is twofold. 

First, it is trivial to see that the sum of weights in the $\ell$-th hop is $\sum_{v\in \V} U(\omega, \rho, \ell)(\D^{-1}\A)^\ell[u,v]=U(\omega, \rho, \ell)$. If $n$ nodes are visited, the average weight is $\Theta(\frac{U(\omega, \rho, \ell)}{n})$, \ie the majority of nodes contribute negligibly to feature aggregations and only a small portion of neighbors with large weights matters. Second, as found in~\cite{XieLYW020, nt2019revisiting}, input data contain not only the low-frequency ground truth but also noises that can originate from falsely labeled data or features. Consequently, incorporating features of those neighbors could potentially incur harmful noises. Therefore, it is a necessity to select important neighbors and filter out insignificant neighbors. 

Based on the above analysis, we aim to identify important neighbors for target node $u$. For ease of exposition, we first define the weight function $\phi(\ell, u, v) =U(\omega, \rho, \ell)\P^\ell[u,v]$ to quantify the importance of neighbor node $v$ to target node $u$, and then formalize the concept of {\em $\varepsilon$-importance neighbor} as follows.

\begin{definition}[$\varepsilon$-Importance Neighbor]
Given a target node $u$ and threshold $\varepsilon \in (0,1)$, node $v$ is called $\varepsilon$-importance neighbor of $u$ if $\exists \ell \in \{0,1,\ldots,\ell_u\}$, we have $\phi(\ell, u, v) \ge \varepsilon$.
\end{definition}

Thanks to the good characteristic of \NDM, a sufficient number of random walks (RWs) are able to identify {\em all} such $\varepsilon$-importance neighbors with high probability, as proved in the following lemma.

\begin{lemma}\label{lem:theta}
Given a target node $u$, threshold $\varepsilon \in (0,1)$, and failure probability $\delta\in (0,1)$, assume $\phi(\ell, u, v) \ge \varepsilon$. Suppose $\theta= \lceil \frac{2\eta^2}{\varepsilon}\log(\frac{1}{\delta \varepsilon}) \rceil$ RWs are generated from $u$ and visit $v$ for $\theta^{(\ell)}_v$ times at the $\ell$-th step. For the weight estimation $\hat{\phi}(\ell, u, v)=\frac{U(\omega, \rho, \ell)\theta^{(\ell)}_v}{\theta}$, we have
\begin{align*}
\Pr\bigg[\underset{0\le \ell\le \ell_u}{\bigvee}\bigg\{\underset{\{v\colon \phi(\ell, u, v) \ge \varepsilon\}}{\bigvee}\Big(\hat{\phi}(\ell, u, v) \le \frac{\phi(\ell, u, v)}{\eta}\Big)\bigg\}\bigg] \le \delta,
\end{align*}
where $\eta>1$ controls the approximation. 
\end{lemma}

\pgfplotsset{compat = newest}

\begin{figure}[!bpt]
\centering
\begin{small}
\begin{tikzpicture}[scale=1,every mark/.append style={mark size=1.5pt}]
    \begin{axis}[
        height=\columnwidth/2.5,
        width=\columnwidth/1.8,
        ylabel= \bf{Weight},
        xlabel={Number of neighbors},
        xmin=0.5, xmax=314.5,
        ymin=0, ymax=1,
        xtick={10,100,200,300},
        yticklabel style = {font=\scriptsize},
        ylabel style={font=\scriptsize},
        ymode=log,
        log basis y={10},
    ] 
    \addplot[line width=0.4mm,color=blue] file[skip first]{data.dat}; 
    \end{axis}
\end{tikzpicture}\hspace{8mm}
\end{small}
 \vspace{-1mm}
\caption{Weight Distribution of Neighbors.} \label{fig:weightdis}
\end{figure}
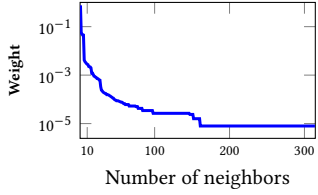

Lemma~\ref{lem:theta} affirms that sufficient RWs could capture $\varepsilon$-importance neighbors with high probability. However, there still contains deficiencies. In particular, along with those $\varepsilon$-important neighbors, many insignificant neighbors will be inevitably selected. For illustration, we randomly choose one target node on dataset Amazon and select its neighbors using RWs. We observe that $10.6\%$ neighbors contribute to $99\%$ weights, and the rest $89.4\%$ neighbors share the left $1\%$ weights, as shown in Figure~\ref{fig:weightdis}. The amount of those insignificant neighbors could unavoidably impair the representation quality. To alleviate the deficiency, we propose to preserve the first-$K$ neighbors with $K=\frac{1}{\varepsilon^2}$.\footnote{One may propose to adopt top-$K$ neighbors. However, top-$K$ selection would incur enormous computation overheads since it requires sorting all neighbors by weights.} 

To explain, in the $\ell$-th hop, each $\varepsilon$-important neighbor will be selected with probability at least $\frac{\varepsilon}{U(\omega, \rho, \ell)}$, and there are at most $\frac{U(\omega, \rho, \ell)}{\varepsilon}$ important neighbors. Thus $\varepsilon$-important neighbors from the $\ell$-th hop will be picked after at most $\frac{U^2(\omega, \rho, \ell)}{\varepsilon^2}$ random selections in expectation. By summing up all $\ell_u$ hops, we have 
\[\sum\nolimits^{\ell_u}_{\ell=0}\tfrac{U^2(\omega, \rho, \ell)}{\varepsilon^2}  \le \sum\nolimits^{\ell_u}_{\ell=0} \tfrac{  U(\omega, \rho, \ell)}{\varepsilon^2} = \tfrac{1}{\varepsilon^2}.\] 

Notice that neither RWs selection nor first-$K$ selection is suitable to solely function as stop conditions. As stated, RW inevitably incurs substantial unimportant neighbors, while first-$K$ selection alone is not bound to terminate when no sufficient neighbors exist. Hence, they compensate each other for better performance. As evaluated in Section~\ref{sec:abla}, first-$K$ selection further boosts the effectiveness notably.

\subsection{Optimized Algorithm \gtro}\label{sec:optimizedalgo}

We propose \gtro in Algorithm~\ref{alg:gtro}, the GCN model by instantiating \NDM. We first initialize the number $\theta$ of RWs according to Lemma~\ref{lem:theta}. Next, we generate length-$\ell_u$ RWs for each $u \in \T$. If neighbor $v$ is visited at the $\ell$-th step, we increase its weight $t_v$ by $\frac{U(\omega, \rho, \ell)}{\theta}$ and store it into set $\S$. This procedure terminates if either the number of RWs reaches $\theta$ or the condition $|\S|\ge \frac{1}{\varepsilon^2}$ is met. Afterward, we update on $\z_u$ (Line~\ref{lineopt:lazyupdate}). Eventually, the final representation $\Z_\T$ is returned once all $|\T|$ target nodes have been processed. 

\begin{algorithm}[!t]
\setlength{\hsize}{0.98\linewidth}
	\caption{\gtro}
	\label{alg:gtro}
	\KwIn{Graph $\G$, feature matrix $\X$, target set $\T$, \newline parameters $\eta$ and $\delta$,  hyperparameters $\tau^\prime$, $\omega$, $\rho$, $\varepsilon$}
	\KwOut{Representation $\Z_\T$}
	$d_\G \gets \frac{2m}{n}$, $\theta\gets \lceil \frac{2\eta^2}{\varepsilon}\log(\frac{1}{\delta \varepsilon}) \rceil$, $\Z_\T \gets \mathbf{0}^{|\T|\times f}$\; 
	\For{$u\in \T$}
	{
	$\ell_u \gets \left \lceil  \tau^\prime\frac{\ln{\frac{2m}{\sqrt{d_{\min}d_u}}}}{\ln{\sqrt{\Delta_\G}}} \right \rceil$\; 
	   \For{$i\leftarrow 1$ \KwTo $\theta$} 
	   {
	      Generate a random walk from node $u$ with length $\ell_u$\;
	      \If{$v$ is visited at the $\ell$-th step}
	      {
	        $t_v \gets t_v+ \frac{U(\omega, \rho, \ell)}{\theta}$\;
	        $\S \gets \S \cup \{v\}$\;
	      }
	      \lIf{$|\S|\ge \frac{1}{\varepsilon^2}$}{\KwBreak}
	   }
	  \lFor{$v\in \S$}{$\z_u \gets \z_u+ t_v\cdot \x_v$\label{lineopt:lazyupdate}} 
	}
	\Return $\Z_\T$\;
\end{algorithm}

Accordingly, for a target node $u$ in $\T$, \gtro is formulated as 
\begin{equation}\label{eqn:gtro}
\textstyle \z_u= \sum_{v\in \{|\textrm{RW}(\N^{(0)}_u \cup \N^{(1)}_u \cup \dotsb \cup \N^{(\ell_u)}_u)|_{\le K}\}} \sum^{\ell_u}_{\ell=0} \frac{U(\omega, \rho, \ell)\theta^{(\ell)}_v}{\theta} \cdot \x_v,    
\end{equation}
where $\textrm{RW}(\N^{(0)}_u \cup \N^{(1)}_u \cup \dotsb \cup \N^{(\ell_u)}_u)$ is the neighbor set identified by RWs within $\ell_u$ hops, $|\cdot|_{\le K}$ indicates the first-$K$ neighbors, and $\sum^{\ell_u}_{\ell=0} \frac{U(\omega, \rho, \ell)\theta^{(\ell)}_v}{\theta}$ is the total estimated weights for selected neighbor $v$.

\spara{Time Complexity} For each node $u$, at most $\theta$ RWs of length $\ell_u$ are generated at the cost of $O(\theta \ell_u)$, and the total number of neighbors is bounded by $O(\theta \ell_u)$, which limits the cost on the feature update to $O(\theta \ell_u f)$. The total cost is $O((f+1)\theta \ell_u)$ for each target node. Let $L=\max\{\ell_u \colon \forall u\in \T\}$. By replacing $\theta=O(\frac{1}{\varepsilon}\log(\frac{1}{\delta \varepsilon}))$, the resulting time complexity of \gtro is $O(\frac{L|\T|f}{\varepsilon}\log\frac{1}{\delta \varepsilon})$.

\subsection{Time Complexity Comparison} 

$\varepsilon$-importance neighbors play a crucial role in feature aggregations. We assume that qualified representations incorporate all $\varepsilon$-importance neighbors with high probability. When capturing such $\varepsilon$-importance neighbors, we analyze and compare the sampling complexities of $9$ representative models\footnote{We assume the models without explicit sampling process with sampling rate $100\%$.} with \gtro, as summarized in Table~\ref{tbl:sampling}. 

\begin{table}[!t]
\caption{Time complexity ($b$ is the batch size, $k$ is the sample size in one hop, $L$ is the propagation length, and $L^\prime$ is the number of model layers).} \label{tbl:sampling}
\begin{tabular}{@{}lll@{}}
\toprule
GNN Method & Preprocessing & Training \\ \midrule
\GraphSAGE  & -  &  $O\big(|\T|k^Lf^2\big)$  \\
\GraphSAINT & -  & $O\big(\frac{Lbmf}{n} + L nf^2\big)$ \\
\FastGCN & - &  $O\big(L|\T|kf + L|\T|f^2\big)$\\
\ShaDowGCN & - &  $O\big(|\T|k^{L} +  L^\prime nf^2\big)$\\
\APPNP & - & $O\big(Lmf + L^\prime nf^2\big)$  \\
\GBP & $O\big(\frac{Lnf\sqrt{L|\T|\log (nL)m/n}}{\varepsilon}\big)$ & $O\big(L^\prime |\T|f^2\big)$ \\
\AGP & $O\big(\frac{L^2nf}{\varepsilon}\big)$ & $O\big(L^\prime|\T|f^2\big)$ \\
\NDLS & $O\big(Lmf\big)$ & $O\big(L^\prime|\T|f^2\big)$ \\
\Grandp & $O\big(\frac{Ln}{\varepsilon}\big)$ & $O\big(knf+L^\prime knf^2\big)$ \\
\gtro & $O\big(\frac{L|\T|f}{\varepsilon}\log\frac{1}{\delta \varepsilon}\big)$ & $O\big(L^\prime|\T|f^2\big)$ \\ \bottomrule
\end{tabular}
\end{table}

\Grandp~\cite{FengDHYCK022} estimates the propagation matrix $\Pi$ during preprocessing with error bounded by $r_\textrm{max}$ at the cost of $O(\frac{(|\U^\prime|+|\T|)L}{r_\textrm{max}})$ where $\U^\prime$ is a sample set of unlabeled node set $\U=\V\setminus \T$. To yield accurate estimations for $\varepsilon$-importance neighbors, $r_\textrm{max}$ is set $r_\textrm{max}=\Theta(\varepsilon)$ and the resulting time complexity of prepossessing is $O(\frac{Ln}{\varepsilon})$. According to~\cite{FengDHYCK022}, its training complexity is $O(knf+L^\prime knf^2)$.

\AGP~\cite{WangHW000W21} provides an accurate estimation for node $u$ if $\pi(u) > \epsilon^\prime$ for each dimension of feature matrix $\X$, denoted as $\x=\X[\cdot, i]$ for $i\in \{0,\ldots,f-1\}$ where $\|\x\|_1 \le 1$, $\pi=\sum^L_{\ell=0}w^{(\ell)}(\D^{-a}\A\D^{-b})^\ell\x$, and $\epsilon^\prime$ is an input threshold. As proved in~\cite{WangHW000W21}, the time cost of \AGP on one feature dimension is $O(\frac{L^2}{\epsilon^\prime}\sum^L_{\ell=0}\|(\sum^\infty_{i=\ell}w^{(i)}) (\D^{-a}\A\D^{-b})^\ell\x \|_1)$. 
We consider a simple case that $a=0$ and $b=1$. Suppose that $\|\x\|_1 =\frac{1}{C}$ for some constant $C\ge 1$, and thus we have \[\left\|\big(\sum^\infty_{i=\ell}w^{(i)}\big) (\A\D^{-1})^\ell\x \right\|_1 =  \sum^\infty_{i=\ell}w^{(i)} \big\|(\A\D^{-1})^\ell \x \big\|_1  = \frac{1}{C}\sum^\infty_{i=\ell}w^{(i)}.\] Since $\sum^L_{\ell=0}\sum^\infty_{i=\ell}w^{(i)} \ge 1$, the time complexity of \AGP is at least $\Theta(\frac{L^2}{C\epsilon^\prime})$. To ensure nodes aggregating at least one $\varepsilon$-importance neighbor $v$ are estimated accurately, $\varepsilon{\x(v)} = \Omega(\epsilon^\prime)$ is required. Since $\|\x\|_1 = \frac{1}{C}$ for some constant $C$ and there are $n$ nodes, it is reasonably to assume that $\x(v)=O(\frac{1}{n})$. Therefore, $\epsilon^\prime = O(\frac{\varepsilon}{n})$. In this regard, the time cost of \AGP to capture $\varepsilon$-importance neighbors for all $f$ dimensions is $O(\frac{L^2nf}{\varepsilon})$.

\GBP~\cite{ChenWDL00W20} derives representations for the $i$-th dimension as $\pi=\sum^L_{\ell=0}w^{(\ell)}\D^{\gamma}(\D^{-1}\A)^\ell \D^{-\gamma} \cdot \x$, where $\x=\X[\cdot, i]$ for $i\in \{0,\ldots,f-1\}$ and $\X$ is the feature matrix with $f$ dimensions. \GBP ensures that the estimation $\hat{\pi}(u)$ of $\pi(u)$ for any $u\in \T$ is within an error of $d^\gamma_u \epsilon^\prime$, \ie $|\pi(u)-\hat{\pi}(u)| \le d^\gamma_u \epsilon^\prime$, where the factor $d^\gamma_u$ is due to the term $\D^{\gamma}$ and does not impact sampling errors. The final time complexity of \GBP is $\frac{Lf\sqrt{L|\T|\log (nL)m/n}}{\epsilon^\prime}$. As discussed above, we have $\varepsilon \x(u) =\Omega(\epsilon^\prime)$ and $\x(u)=O(\frac{1}{n})$, which indicates that $\epsilon^\prime = O(\frac{\varepsilon}{n})$. Consequently, the time cost of \GBP to capture $\varepsilon$-importance neighbors is $O\big(\frac{Lnf\sqrt{L|\T|\log (nL)m/n}}{\varepsilon}\big)$.

For the rest models in Table~\ref{tbl:sampling}, we borrow the time complexity from their official analyses since they either provide no sampling approximation guarantee or consider all neighbors without explicit sampling. As analyzed, time complexities of state of the art are linear in the size of the graph, while that of \gtro is linear in the size of the target set $\T$. In semi-supervised classification with limited labels, we have $|\T| \ll n$, which confirms the theoretical efficiency superiority of \gtro.

\spara{Parallelism} \gtro derives the representation of every target node {\em independently} and does not rely on any intermediate representations of other nodes. This design makes \gtro inherently parallelizable so as to be a promising solution to derive node representations for massive graphs since they can process all nodes simultaneously. Further, this enables \gtro scalable for supervised learning as well.

\vspace{-1mm}
\section{Experiment}\label{sec:exp}
In this section, we evaluate the performance of \gtro for semi-supervised classification in terms of effectiveness (micro F1-scores) and efficiency (running times).

\subsection{Experimental Setting}\label{sec:setting}
\vspace{-1.5mm}
\spara{Datasets} We use seven publicly available datasets across various sizes in our experiments. Specifically, we conduct {\em transductive learning} on the four citation networks, including three small citation networks~\cite{sen2008collective} Cora, Citeseer, and Pubmed, and a web-scale citation network Papers100M~\cite{hu2020open}. We run {\em inductive learning} on three large datasets, \ie citation network Ogbn-arxiv~\cite{hu2020open}, social network Reddit~\cite{ZengZSKP20}, and co-purchasing network Amazon~\cite{ZengZSKP20}. Table~\ref{tbl:dataset} in Appendix~\ref{app:setting} summarizes the statistics of those datasets. Among them, Papers100M is the largest dataset ever tested in the literature.

For semi-supervised classification with limited labels, we randomly sample $20$ nodes per class for training, $500$ nodes for validation, and $1000$ nodes for testing. For each dataset, we randomly generate $10$ instances and report the average performance of each tested method.

\spara{Baselines} For transductive learning, we evaluate \gtro against $13$ baselines. We categorize them into three types,  \ie (i) $2$ {\em coupled} GNN methods \GCN~\cite{KipfW17} and \GAT~\cite{VelickovicCCRLB18}, (ii) $3$ {\em sampling-based} methods \GraphSAGE~\cite{HamiltonYL17}, \GraphSAINT~\cite{ZengZSKP20}, and \ShaDowGCN~\cite{ZengZXSMKPJC21}, and (iii) $8$ {\em decoupled} GNN methods \SGC~\cite{WuSZFYW19}, \APPNP~\cite{KlicperaBG19}, and its improvement \PPRGo~\cite{BojchevskiKPKBR20}, \GDC~\cite{KlicperaWG19}, \GBP~\cite{ChenWDL00W20}, \AGP~\cite{WangHW000W21}, and two recently proposed \NDLS~\cite{ZhangYSLOTYC21}, and \Grandp~\cite{FengDHYCK022}. 

For inductive learning, we compare \gtro with $7$ baselines. Among the $13$ methods tested in transductive learning, $7$ of them are not suitable for semi-supervised inductive learning and thus are omitted, as explained in Section~\ref{app:setting}. In addition, we include an extra method \FastGCN~\cite{ChenMX18} designed for inductive learning. Details for the implementations are provided in Appendix~\ref{app:setting}.

\spara{Parameter Settings} For \gtro, we fix $\eta=2$, $\delta=0.01$ and tune the four hyperparameters $\tau^\prime$, $\omega$, $\rho$, and $\varepsilon$. Appendix~\ref{app:setting} provides the principal on how they are tuned and values selected for all datasets. As with baselines, we either adopt their suggested parameter settings or tune the parameters following the same principle as \gtro to reach their best possible performance.

All methods are evaluated in terms of {\em micro F1-scores} on node classification and {\em running times} including preprocessing times (if applicable) and training times. One method is omitted on certain datasets if it (i) is not suitable for inductive semi-supervised learning or (ii) runs out of memory (OOM), either GPU memory or RAM.

\begin{table}[!t]
 \caption{F1-score (\%) of transductive learning.} \label{tbl:transductive}
 \setlength{\tabcolsep}{0.2em} 
 \small
\begin{tabular}{@{}clcccc@{}}
\toprule
 & {Methods} & {Cora} & {Citeseer} & {Pubmed} & {Papers100M}  \\ \midrule 
\multirow{2}{*}{\rotatebox[origin=c]{90}{Coup*}}
& \GCN	    	&		78.91 $\pm$ 1.87			&		69.11 $\pm$ 1.46			&		78.05 $\pm$ 1.64			&	OOM	\\
& \GAT            & 80.22 $\pm$ 1.48                  & 69.35 $\pm$ 0.93                  & 78.82 $\pm$ 1.90 & OOM \\ \midrule
\multirow{3}{*}{\rotatebox[origin=c]{90}{Sampling}}
& \GraphSAGE      &  76.67 $\pm$ 2.21                 & 67.41 $\pm$ 1.77                  & 76.92 $\pm$ 2.84 &  OOM \\
& \GraphSAINT     &  74.76 $\pm$ 2.86                 & 67.51 $\pm$ 4.76                  & 78.65 $\pm$ 4.17 & OOM \\
& \ShaDowGCN      & 73.83 $\pm$ 2.46                  &  63.54 $\pm$ 1.11                 & 71.79 $\pm$ 2.92 & OOM \\ \midrule
\multirow{9}{*}{\rotatebox[origin=c]{90}{Decoupled}}
& \SGC	    	& 76.79 $\pm$ 1.82	&	70.49 $\pm$ 1.29	&	74.11 $\pm$ 2.55			&	\underline{48.59 $\pm$ 1.77}	\\
& \APPNP		&	81.16 $\pm$ 0.77			    &		69.83 $\pm$ 1.27			&	80.21 $\pm$ 1.79		&	OOM	\\
& \PPRGo		    &		79.01 $\pm$ 1.88			&		68.92 $\pm$  1.72			&		78.20 $\pm$ 1.96			&	OOM	\\
& \GDC	    	&		80.69 $\pm$  1.99			&	 69.69 $\pm$  1.42			    &		77.67 $\pm$  1.65			&	OOM	\\
& \GBP	    	&		81.17 $\pm$ 1.60			&		70.18 $\pm$  1.90			&		80.09 $\pm$  1.51	&	44.91 $\pm$  1.23	\\
& \AGP	    	&		77.70 $\pm$ 2.04			&	  67.15 $\pm$   2.04			&		78.97 $\pm$ 1.33	&	46.71 $\pm$  1.99	\\
& \NDLS       & 81.39 $\pm$ 1.55                      & 69.63 $\pm$ 1.69                  & \underline{80.38 $\pm$ 1.41} & OOM\\
& \Grandp		& \bf{83.48 $\pm$ 1.18}			   &	\bf{71.42 $\pm$ 1.89}	            &	79.18 $\pm$ 1.93		&	OOM	\\
& \gtro		&		\underline{82.13 $\pm$ 1.08}		 &		\underline{71.35  $\pm$ 0.82}		&		\bf{80.90 $\pm$ 2.02}		&	\bf{49.81 $\pm$  1.10}	\\ \bottomrule
\end{tabular}
\end{table}

\subsection{Performance Results}\label{sec:performance}

Table~\ref{tbl:transductive} and Table~\ref{tbl:inductive} present the averaged F1-scores associated with the standard deviations in {\em transductive learning} on Cora, Citeseer, Pubmed, and Papers100M and {\em inductive learning} on Ogbn-arxiv, Reddit, and Amazon respectively. For ease of demonstration, we highlight the {\em largest} score in bold and underline the {\em second largest} score for each dataset. 

Table~\ref{tbl:transductive} shows that \gtro achieves the highest F1-scores on datasets Pubmed and Papers100M and the second highest scores on Cora and Citeseer.  Meanwhile, \gtro obtains the largest F1-scores on the three datasets Ogbn-arxiv, Reddit, and Amazon, as displayed in Table~\ref{tbl:inductive}. In particular, the improvement margins over the second best on the three datasets are $0.93\%$, $0.78\%$, and $2.67\%$ respectively. These observations indicate that \gtro performs better on relatively large graphs. Intuitively, nodes in large graphs are prone to reside in various structure contexts and contain neighbors of mixed quality, and the \NDM diffusion model and the sampling techniques (important neighbor identification and selection) utilized by \gtro are able to take advantage of such node-wise characteristics.

The most competitive method, \Grandp achieves the best on datasets Cora and Citeseer. Nonetheless, as shown in Figure~\ref{fig:runtimesmall} ( Section~\ref{app:exp}), \Grandp runs significantly slower than \gtro does. For the three sampling-based methods, \ie \GraphSAGE, \GraphSAINT, and \ShaDowGCN, they acquire noticeably lower F1-scores than \gtro does. This is due to that they sample neighbors and nodes randomly without customizing the sampling strategy towards target nodes, as introduced in Section~\ref{sec:relatedwork}. Meanwhile, the clear performance improvement of \gtro over \GBP and \AGP clearly supports the superiority of our general heat diffusion function \GHD over the diffusion models used in \GBP and \AGP (i.e., PPR and HKPR), as well as the efficacy of our diffusion model \NDM.

Overall, it is crucial to consider the unique structure characteristic of each individual node in the design of both the diffusion model and neighbor sampling techniques for node classifications.

\begin{table}[!t]
 \caption{F1-score (\%) of inductive learning.} \label{tbl:inductive} 
 \small
\begin{tabular}{@{}lccc@{}}
\toprule
Methods & Ogbn-arxiv & Reddit & Amazon  \\ \midrule
\GraphSAGE  & 51.79 $\pm$ 2.16 & 89.16 $\pm$ 1.16 &  47.71 $\pm$ 1.07 \\
\FastGCN & \underline{56.45 $\pm$ 1.69} & 92.43 $\pm$ 1.00 & OOM \\
\SGC     	&		56.03 $\pm$ 1.96		&	 \underline{92.64 $\pm$ 1.02}		&	   41.32 $\pm$ 1.10		\\
\PPRGo   	&		52.12 $\pm$ 3.22		&		78.21 $\pm$  3.07		&		60.10 $\pm$  1.17		\\
\GBP     	&		54.01 $\pm$ 2.55		&		76.09 $\pm$  1.75		&		\underline{60.78 $\pm$  1.04}		\\
\AGP     	&	    55.89 $\pm$ 1.47		&		92.18 $\pm$ 0.88		&		55.72 $\pm$  1.68		\\
\NDLS       & 54.23 $\pm$ 2.49 & 85.25 $\pm$ 1.24 &  50.10 $\pm$ 2.09 \\
\gtro   	&		\bf{57.38 $\pm$ 1.31}		&		\bf{93.42 $\pm$ 0.48}	&		\bf{63.45 $\pm$  0.70}	\\\bottomrule
\end{tabular}
\end{table}


\def\st{0.83}
\def\gap{0.048}
\def\wid{1.7}

\begin{figure*}[!t]
\centering
\begin{small}
\begin{tikzpicture}
    \begin{customlegend}[
        legend entries={{\bf \GraphSAGE},{\bf \FastGCN},{\bf \SGC},{\bf \PPRGo},{\bf \GBP},{\bf \AGP},{\bf \NDLS},{\bf \gtro}},
        legend columns=8,
        area legend,
        legend style={at={(0.45,1.15)},anchor=north,draw=none,font=\footnotesize,column sep=0.15cm}]
        \addlegendimage{,fill=mygreen}
        \addlegendimage{,fill=blue!70}
        \addlegendimage{,pattern color=blue,pattern=crosshatch dots}
        \addlegendimage{,pattern color=blue,pattern=grid}
        \addlegendimage{,pattern color=myred,pattern=north west lines}
        \addlegendimage{,pattern color=magenta, pattern=crosshatch}
        \addlegendimage{,pattern color=red,pattern=horizontal lines}
        \addlegendimage{,fill=red}
    \end{customlegend}
\end{tikzpicture}\vspace{-2mm}
\subfloat[Ogbn-arxiv]{
\begin{tikzpicture}[scale=1]
\begin{axis}[
    height=\columnwidth/2.5,
    width=\columnwidth/\wid,
    ybar=1.5pt,
    bar width=0.35cm,
    enlarge x limits=true,
    ylabel={\em running time} (sec),
    xticklabel=\empty,
    ymin=0.1,
    ymax=100,
    ymode=log,
    ytick={0.1,1,10,100},
    log origin y=infty,
    log basis y={10},
    every axis y label/.style={at={(current axis.north west)},right=10mm,above=0mm},
    ]

\addplot [,fill=mygreen] coordinates {(1,1.9185) }; 
\addplot [,fill=blue!70] coordinates {(1,22.6958)}; 
\addplot [,pattern color=blue,pattern=crosshatch dots] coordinates {(1,0.7757)}; 
\addplot [,pattern color=blue,pattern=grid] coordinates {(1,43.3253)}; 
\addplot [,pattern color=myred,pattern=north west lines] coordinates {(1,31.16436)}; 
\addplot [,pattern color=magenta, pattern=crosshatch] coordinates {(1,1.107592)}; 
\addplot [,pattern color=red,pattern=horizontal lines] coordinates {(1,3.0274)}; 
\addplot [,fill=red] coordinates {(1,1.18926)}; 

\end{axis}
\end{tikzpicture}\vspace{-2mm}\hspace{0mm}\label{fig:arxivS4}%
}%
\subfloat[Reddit]{
\begin{tikzpicture}[scale=1]
\begin{axis}[
    height=\columnwidth/2.5,
    width=\columnwidth/\wid,
    ybar=1.5pt,
    bar width=0.35cm,
    enlarge x limits=true,
    ylabel={\em running time} (sec),
    xticklabel=\empty,
    ymin=1,
    ymax=1000,
    ymode=log,
    ytick={1,10,100,1000},
    log origin y=infty,
    log basis y={10},
    every axis y label/.style={at={(current axis.north west)},right=10mm,above=0mm},
    ]

\addplot [,fill=mygreen] coordinates {(1,34.5487) }; 
\addplot [,fill=blue!70] coordinates {(1,43.5498)}; 
\addplot [,pattern color=blue,pattern=crosshatch dots] coordinates {(1,10.1326)}; 
\addplot [,pattern color=blue,pattern=grid] coordinates {(1,321.13689)}; 
\addplot [,pattern color=myred,pattern=north west lines] coordinates {(1,307.4395)}; 
\addplot [,pattern color=magenta, pattern=crosshatch] coordinates {(1,8.90624)}; 
\addplot [,pattern color=red,pattern=horizontal lines] coordinates {(1,44.55)}; 
\addplot [,fill=,fill=red] coordinates {(1,2.163939)}; 

\end{axis}
\end{tikzpicture}\vspace{-2mm}\hspace{0mm}\label{fig:redditS4}%
}%
\subfloat[Amazon]{
\begin{tikzpicture}[scale=1]
\begin{axis}[
    height=\columnwidth/2.5,
    width=\columnwidth/\wid,
    ybar=1.5pt,
    bar width=0.35cm,
    enlarge x limits=true,
    ylabel={\em running time} (sec),
    xticklabel=\empty,
    ymin=1,
    ymax=100,
    ymode=log,
    ytick={1,10,100},
    log origin y=infty,
    log basis y={10},
    every axis y label/.style={at={(current axis.north west)},right=10mm,above=0mm},
    ]

\addplot [,fill=mygreen] coordinates {(1,49.4441) }; 
\addplot [,pattern color=blue,pattern=crosshatch dots] coordinates {(1,25.3175)}; 
\addplot [,pattern color=blue,pattern=grid] coordinates {(1,22.01036)}; 
\addplot [,pattern color=myred,pattern=north west lines] coordinates {(1,11.29867)}; 
\addplot [,pattern color=magenta, pattern=crosshatch] coordinates {(1,23.5255)}; 
\addplot [,pattern color=red,pattern=horizontal lines] coordinates {(1,44.55)}; 
\addplot [,fill=red] coordinates {(1,1.2695)}; 

\end{axis}
\end{tikzpicture}\vspace{-2mm}\hspace{0mm}\label{fig:amazonS4}%
}%
\subfloat[Papers100M]{
\begin{tikzpicture}[scale=1]
\begin{axis}[
    height=\columnwidth/2.5,
    width=\columnwidth/1.9,
    ybar=1.5pt,
    bar width=0.4cm,
    enlarge x limits=true,
    ylabel={\em running time} (sec),
    xticklabel=\empty,
    ymin=1,
    ymax=10000,
    ymode=log,
    ytick={1,10,100,1000,10000},
    log origin y=infty,
    log basis y={10},
    every axis y label/.style={at={(current axis.north west)},right=10mm,above=0mm},
    ]

\addplot [,pattern color=blue,pattern=crosshatch dots] coordinates {(1,2355.808)}; 
\addplot [,pattern color=myred,pattern=north west lines] coordinates {(1,2441.7021)}; 
\addplot [,pattern color=magenta, pattern=crosshatch] coordinates {(1,2618.8565)}; 
\addplot [,fill=red] coordinates {(1,5.33461) }; 

\end{axis}
\end{tikzpicture}\vspace{-2mm}\hspace{0mm}\label{fig:papers100MS4}%
}%

\end{small}
\caption{Running times on large graphs (best viewed in color).} \label{fig:runtime}
\vspace{-1mm}
\end{figure*}
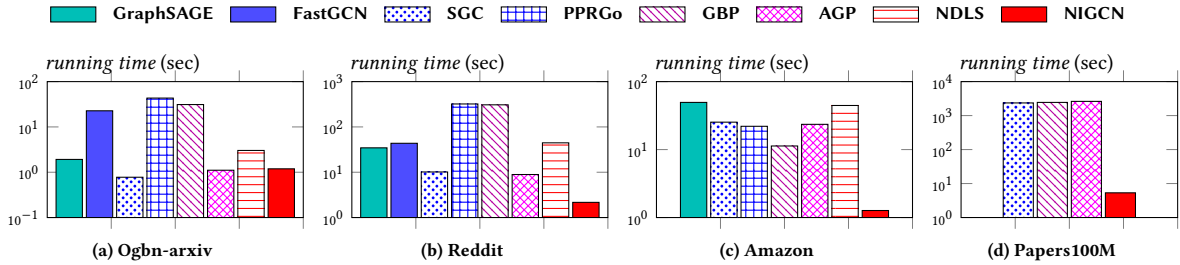

\subsection{Scalability Evaluation on Large Graphs}\label{sec:efficiency}
 
In this section, we evaluate the scalability of tested methods by comparing their running times on the four large datasets, Ogbn-arxiv, Reddit, Amazon, and Papers100M. In particular, the running times include preprocessing times (if applicable) and training times. For a comprehensive evaluation, we also report the corresponding running times on the three small datasets Cora, Citeseer, and Pubmed in Appendix~\ref{app:exp}.

As shown in Figure~\ref{fig:runtime}, \gtro ranks third with negligible lags on dataset Ogbn-arxiv and dominates other methods noticeably on datasets Reddit, Amazon, and Papers100M. Meanwhile, its efficiency advantage expands larger as datasets grow. Specifically, on dataset Ogbn-arxiv, \gtro, \SGC, and \AGP all can finish running within $1$ second. The speedups of \gtro against the second best on datasets Reddit, Amazon, and Papers100M are up to $4.12\times$, $8.90\times$, and $441.61\times$ respectively. In particular, on the largest dataset Papers100M, \gtro is able to complete preprocessing and model training within $10$ seconds. The remarkable scalability of \gtro lies in that (i) \gtro only generates node representations for a small portion of labeled nodes involved in model training, and (ii) neighbor sampling techniques in \gtro significantly reduce the number of neighbors for each labeled node in feature aggregation, as detailed analyzed in Section~\ref{sec:instant}. This observation strongly supports the outstanding scalability of \gtro and its capability to handle web-scale graphs.

\subsection{Ablation Study}\label{sec:abla}

\begin{table}[!t]
 \caption{Performance of \gtro variants.} \label{tbl:ablation} 
 \small
\begin{tabular}{@{}lcccc@{}}
\toprule
{Variants} & {$\gtro_{\textsf{HKPR}}$} & {$\gtro_{\textsf{UDL}}$} & {$\gtro_{\textsf{NFK}}$} & {\gtro} \\ \midrule
F1-score (\%)  & 61.60 $\pm$ 0.82 & 61.32 $\pm$ 1.32 & 52.76 $\pm$ 1.14 & 63.45 $\pm$ 0.70\\	   
Disparity (\%) & -1.85		& -2.13		&	-10.69	 & 0\\ \bottomrule
\end{tabular}
\end{table}

\spara{Variants of \gtro} To validate the effects of diffusion model \NDM and the sampling techniques in \gtro, we design three variants of \gtro, \ie $\gtro_{\textsf{HKPR}}$, $\gtro_{\textsf{UDL}}$, and $\gtro_{\textsf{NFK}}$.
Specifically, (i) $\gtro_{\textsf{HKPR}}$ adopts heat kernel PageRank (HKPR) instead of general heat diffusion \GHD in \NDM as the diffusion function, (ii) $\gtro_{\textsf{UDL}}$ unifies the diffusion length for all labeled nodes in contrast with the node-wise diffusion length in \NDM, and (iii) $\gtro_{\textsf{NFK}}$ removes the first-$K$ limitation on the number of neighbors. 
We test all variants on Amazon and Table~\ref{tbl:ablation} reports the corresponding F1-scores. For clarification, we also present the F1-score disparity from that of \gtro.

First of all, we observe that the F1-score of $\gtro_{\textsf{HKPR}}$ is $1.04\%$ smaller than that of \gtro. This verifies that HKPR is not capable of capturing the structure characteristics of Amazon and \NDM offers better generality. Second, $\gtro_{\textsf{UDL}}$ acquires  $1.32\%$ less F1-scores compared with \gtro. This suggests that diffusion with customized length leverages the individual structure property of each target node, which benefits the node classification. Last,  $\gtro_{\textsf{NFK}}$ achieves $9.88\%$ smaller F1-score than \gtro does, which reveals the potential noise signals from neighbors and recognizes the importance and necessity of important neighbor selection.

\begin{table}[!t]
 \caption{F1-score (\%) vs.\ label percentages.} \label{tbl:varyinglabel} 
 \small
 \resizebox{0.48\textwidth}{!}{%
\begin{tabular}{@{}lccccc@{}}
\toprule
Methods & 4\textperthousand & 8\textperthousand & 1\% & 2\% & 5\% \\ \midrule
\PPRGo   &	63.68 $\pm$ 1.13	&	66.02 $\pm$ 0.89	&	66.62 $\pm$  0.86	&	67.54 $\pm$ 0.67    &	OOM \\
\GBP     &	\underline{64.57 $\pm$ 0.66}	&	\underline{68.31 $\pm$ 0.86}	&	\underline{69.22 $\pm$  0.71}	&	\underline{71.56 $\pm$ 0.54}	&	\underline{73.57 $\pm$ 0.45}    \\
\gtro    &	\bf{67.24 $\pm$ 0.60}	&	\bf{70.93 $\pm$ 0.75}	&	\bf{71.49 $\pm$  0.91}	&	\bf{73.19 $\pm$ 0.52}	&	\bf{74.02 $\pm$ 0.31} \\ \bottomrule
\end{tabular}}
\end{table}

\spara{Label Percentages Varying} To evaluate the robustness of \gtro towards the portion of labeled nodes, we test \gtro by varying the label percentages in $\{4\text{\textperthousand}, 8\text{\textperthousand}, 1\%, 2\%, 5\%\}$ on Amazon\footnote{Papers100M contains only $1.4\%$ labeled nodes which is insufficient for testing.} and compare it with two competitive baselines \PPRGo and \GBP. Results in Table~\ref{tbl:varyinglabel} and Figure~\ref{fig:varyinglabel} report the F1-scores and running times respectively.

As displayed in table~\ref{tbl:varyinglabel}, \gtro achieves the highest F1-score with average $1.93\%$ advantages over the second highest scores across tested percentage ranges. Moreover, Figure~\ref{fig:varyinglabel} shows that \gtro notably dominates the other two competitors in efficiency. In particular, \gtro completes execution within $3$ seconds and runs up to $5\times$ to $20\times$ faster than \GBP and \PPRGo in {\em all} settings respectively. These findings validate the robustness and outstanding performance of \gtro for semi-supervised classification. 

\spara{Parameter Analysis} The performance gap between $\gtro_{\textsf{HKPR}}$ and \gtro has shown that inappropriate combination of $\omega$ and $\rho$ degrades the performance significantly. Here we test the effects of hyperparameters $\tau$ and $\varepsilon$ in control of the diffusion length $\ell_u$ and the number of neighbors, respectively.

On Amazon, we select $\tau=1.5$, denoted as $\tau_o$ and $\varepsilon=0.05$, denoted as $\varepsilon_o$. We then test $\tau \in \{0.25\tau_o, 0.5\tau_o, 2\tau_o, 4\tau_o\}$ and $\varepsilon \in \{0.25\varepsilon_o, 0.5\varepsilon_o, 2\varepsilon_o, 4\varepsilon_o\}$ and plot the results in Figure~\ref{fig:varyingtaueps}. As shown, F1-score  improves along with the increase of $\tau$ until $\tau=\tau_o$ and then decreases slightly as expected. Similar patterns are also observed in the case of $\varepsilon$. Specifically, \gtro exhibits more sensitivity towards the change of $\tau$ than that of $\varepsilon$. This is because \gtro is able to capture the most important neighbors within the right $\tau$-distance with high probability when changing the threshold of $\varepsilon$-importance neighbors, which, however, is not guaranteed when altering the bound of $\tau$-distance.

\begin{figure}
\centering
\begin{minipage}{0.5\linewidth}
\centering
\begin{small}
\begin{tikzpicture}[scale=1,every mark/.append style={mark size=1.5pt}]
    \begin{axis}[
        height=\columnwidth/2.5*2,
        width=\columnwidth/2*2.3,
        ylabel={\em running time} (sec),
        xmin=0.5, xmax=5.5,
        ymin=1, ymax=10000,
        xtick={1,2,3,4,5},
        xticklabel style = {font=\footnotesize},
        yticklabel style = {font=\footnotesize},
        xticklabels={4\text{\textperthousand}, 8\text{\textperthousand}, 1\%, 2\%, 5\%},
        ymode=log,
        log basis y={10},
        every axis y label/.style={font=\footnotesize,at={(current axis.north west)},right=9mm,above=0mm},
        legend columns=2,
        legend style={fill=none,font=\scriptsize,at={(0.5,0.75)},anchor=center,draw=none},
    ] 
    \addplot[line width=0.2mm,mark=otimes,color=orange]
    plot coordinates {
    (1,	35.87547)
    (2,	80.36992)
    (3,	98.21063)
    (4,	207.08941)
    }; 
    \addplot[line width=0.2mm,mark=triangle,color=blue]
    plot coordinates {
    (1,	10.15739)
    (2,	11.36285)
    (3,	11.02694)
    (4,	12.12547)
    (5,	12.87104)
    }; 
    \addplot[line width=0.2mm,mark=o,color=red]
    plot coordinates {
    (1,	1.2795)
    (2,	1.843068)
    (3,	1.833078)
    (4,	2.136608)
    (5,	2.56482)
    }; 
    \legend{\PPRGo, \GBP, \gtro}
    \end{axis}
\end{tikzpicture}
\end{small}
\captionof{figure}{Running time.} \label{fig:varyinglabel}
\end{minipage}%
\begin{minipage}{0.5\linewidth}
\centering
\begin{small}
\begin{tikzpicture}[scale=1,every mark/.append style={mark size=1.5pt}]
    \begin{axis}[
        height=\columnwidth/2.5*2,
        width=\columnwidth/2*2.3,
        ylabel={\em F1-scores} (\%),
        xmin=0.5, xmax=5.5,
        ymin=55, ymax=74,
        xtick={1,2,3,4,5},
        xticklabel style = {font=\footnotesize},
        yticklabel style = {font=\footnotesize},
        xticklabels={0.25,0.5,1,2,4},
        every axis y label/.style={font=\footnotesize,at={(current axis.north west)},right=9mm,above=0mm},
        legend columns=1,
        legend style={fill=none,font=\scriptsize,at={(0.35,0.75)},anchor=center,draw=none},
    ] 
    \addplot[line width=0.2mm,mark=triangle,color=blue]
    plot coordinates {
    (1,	56.94)
    (2,	61.89)
    (3,	63.45)
    (4,	63.40)
    (5,	63.37)
    }; 
    \addplot[line width=0.2mm,mark=o,color=red]
    plot coordinates {
    (1,	61.42)
    (2,	62.46)
    (3,	63.45)
    (4,	63.36)
    (5,	63.29)
    }; 
    \legend{\gtro ($\tau$), \gtro ($\varepsilon$)}
    \end{axis}
\end{tikzpicture}
\end{small}
\captionof{figure}{Performance.} \label{fig:varyingtaueps}
\end{minipage}
\end{figure}

\section{Conclusion}\label{sec:conclusion}

In this paper, we propose \gtro, a scalable graph neural network built upon the node-wise diffusion model \NDM, which achieves orders of magnitude speedups over representative baselines on massive graphs and offers the highest F1-score on semi-supervised classification. 
In particular, \NDM (i) utilizes the individual topological characteristic and yields a unique diffusion scheme for each target node and (ii) adopts a general heat diffusion function \GHD that adapts well to various graphs. Meanwhile, to optimize the efficiency of feature aggregations, \gtro computes representations for target nodes only and leverages advanced neighbor sampling techniques to identify and select important neighbors, which not only improves the performance but also boosts the efficiency significantly. Extensive experimental results strongly support the state-of-the-art performance of \gtro for semi-supervised classification and the remarkable scalability of \gtro.

%% file: appendix.tex
\def\st{0.82}
\def\gap{0.0276}
\begin{figure*}[!t]
\centering
\begin{small}
\begin{tikzpicture}
    \begin{customlegend}[
        legend entries={{\bf \GCN},{\bf \GAT},{\bf \GraphSAGE},{\bf \GraphSAINT},{\bf \ShaDowGCN},{\bf \SGC},{\bf \APPNP}},
        legend columns=7,
        area legend,
        legend style={at={(0.45,1.15)},anchor=north,draw=none,font=\footnotesize,column sep=0.15cm}]
        \addlegendimage{,fill=white}
        \addlegendimage{,fill=myorange}
        \addlegendimage{,fill=mygreen}
        \addlegendimage{,fill=cyan}
        \addlegendimage{,fill=blue!70}
        \addlegendimage{,pattern color=blue,pattern=crosshatch dots}
        \addlegendimage{,pattern color=blue!70,pattern=crosshatch}
    \end{customlegend}
    \begin{customlegend}[
        legend entries={{\bf \PPRGo},{\bf \GDC},{\bf \GBP},{\bf \AGP},{\bf \NDLS},{\bf \Grandp},{\bf \gtro}},
        legend columns=7,
        area legend,
        legend style={at={(0.45,0.76)},anchor=north,draw=none,font=\footnotesize,column sep=0.15cm}]
        \addlegendimage{,pattern color=blue,pattern=grid}
        \addlegendimage{,pattern color=myredd,pattern=north east lines}
        \addlegendimage{,pattern color=myred,pattern=north west lines}
        \addlegendimage{,pattern color=magenta, pattern=crosshatch}
        \addlegendimage{,pattern color=red,pattern=horizontal lines}
        \addlegendimage{,pattern color=red,pattern=grid}
        \addlegendimage{,fill=red}
    \end{customlegend}
\end{tikzpicture}\vspace{-2mm}
\subfloat[Cora]{
\begin{tikzpicture}[scale=1]
\begin{axis}[
    height=\columnwidth/2,
    width=\columnwidth/1.3,
    ybar=1.5pt,
    bar width=0.29cm,
    enlarge x limits=true,
    ylabel={\em running time} (sec),
    xticklabel=\empty,
    ymin=0.1,
    ymax=100,
    ymode=log,
    ytick={0.1,1,10,100},
    log origin y=infty,
    log basis y={10},
    every axis y label/.style={at={(current axis.north west)},right=10mm,above=0mm},
    legend style={at={(0.02,0.98)},anchor=north west,cells={anchor=west},font=\tiny}]
\addplot [,fill=white] coordinates {(1,0.6122)}; 
\addplot [,fill=myorange] coordinates {(1,0.9842)}; 
\addplot [,fill=mygreen] coordinates {(1,8.3396)}; 
\addplot [,fill=cyan] coordinates {(1,16.6031)}; 
\addplot [,fill=blue!70] coordinates {(1,9.1415)}; 
\addplot [,pattern color=blue,pattern=crosshatch dots] coordinates {(1,0.6941)}; 
\addplot [,pattern color=blue!70,pattern=crosshatch] coordinates {(1,1.2862)}; 
\addplot [,pattern color=blue,pattern=grid] coordinates {(1,10.4069)}; 
\addplot [,pattern color=myredd,pattern=north east lines] coordinates {(1,2.9002)}; 
\addplot [,pattern color=myred,pattern=north west lines] coordinates {(1,0.696)}; 
\addplot [,pattern color=magenta, pattern=crosshatch] coordinates {(1,0.519)}; 
\addplot [,pattern color=red,pattern=horizontal lines] coordinates {(1,6.4794)}; 
\addplot [,pattern color=red,pattern=grid] coordinates {(1,45.316)};
\addplot [,fill=red] coordinates {(1,3.1804)};

\end{axis}
\end{tikzpicture}\vspace{-2mm}\hspace{0mm}\label{fig:arxivtime}%
}%
\hspace{-1mm}
\subfloat[Citeseer]{
\begin{tikzpicture}[scale=1]
\begin{axis}[
    height=\columnwidth/2,
    width=\columnwidth/1.3,
    ybar=1.5pt,
    bar width=0.29cm,
    enlarge x limits=true,
    ylabel={\em running time} (sec),
    xticklabel=\empty,
    ymin=0.1,
    ymax=100,
    ymode=log,
    ytick={0.1,1,10,100},
    log origin y=infty,
    log basis y={10},
    every axis y label/.style={at={(current axis.north west)},right=10mm,above=0mm},
    ]

\addplot [,fill=white] coordinates {(1,0.6076)}; 
\addplot [,fill=myorange] coordinates {(1,1.3952)}; 
\addplot [,fill=mygreen] coordinates {(1,14.2971)}; 
\addplot [,fill=cyan] coordinates {(1,25.1651)}; 
\addplot [,fill=blue!70] coordinates {(1,9.117)}; 
\addplot [,pattern color=blue,pattern=crosshatch dots] coordinates {(1,2.7186)}; 
\addplot [,pattern color=blue!70,pattern=crosshatch] coordinates {(1,1.2527)}; 
\addplot [,pattern color=blue,pattern=grid] coordinates {(1,4.362)}; 
\addplot [,pattern color=myredd,pattern=north east lines] coordinates {(1,3.9888)}; 
\addplot [,pattern color=myred,pattern=north west lines] coordinates {(1,1.3139)}; 
\addplot [,pattern color=magenta, pattern=crosshatch] coordinates {(1,2.824827)}; 
\addplot [,pattern color=red,pattern=horizontal lines] coordinates {(1,19.9362)}; 
\addplot [,pattern color=red,pattern=grid] coordinates {(1,64.68)};
\addplot [,fill=red] coordinates {(1,3.3589)};

\end{axis}
\end{tikzpicture}\vspace{-2mm}\hspace{0mm}\label{fig:amazontime}%
}%
\hspace{-1mm}
\subfloat[Pubmed]{
\begin{tikzpicture}[scale=1]
\begin{axis}[
    height=\columnwidth/2,
    width=\columnwidth/1.3,
    ybar=1.5pt,
    bar width=0.29cm,
    enlarge x limits=true,
    ylabel={\em running time} (sec),
    xticklabel=\empty,
    ymin=0.1,
    ymax=100,
    ymode=log,
    ytick={0.1,1,10,100},
    log origin y=infty,
    log basis y={10},
    every axis y label/.style={at={(current axis.north west)},right=10mm,above=0mm},
    ]

\addplot [,fill=white] coordinates {(1,1.9173)}; 
\addplot [,fill=myorange] coordinates {(1,6.8416)}; 
\addplot [,fill=mygreen] coordinates {(1,3.5329)}; 
\addplot [,fill=cyan] coordinates {(1,23.0627)}; 
\addplot [,fill=blue!70] coordinates {(1,22.1788)}; 
\addplot [,pattern color=blue,pattern=crosshatch dots] coordinates {(1,2.6265)}; 
\addplot [,pattern color=blue!70,pattern=crosshatch] coordinates {(1,1.3135)}; 
\addplot [,pattern color=blue,pattern=grid] coordinates {(1,6.4469)}; 
\addplot [,pattern color=myredd,pattern=north east lines] coordinates {(1,25.0988)}; 
\addplot [,pattern color=myred,pattern=north west lines] coordinates {(1,1.266)}; 
\addplot [,pattern color=magenta, pattern=crosshatch] coordinates {(1,0.688986)}; 
\addplot [,pattern color=red,pattern=horizontal lines] coordinates {(1,24.5738)}; 
\addplot [,pattern color=red,pattern=grid] coordinates {(1,20.7)};
\addplot [,fill=red] coordinates {(1,2.6189)};

\end{axis}
\end{tikzpicture}\vspace{-2mm}\hspace{0mm}\label{fig:paperstime}%
}%
\end{small}
\caption{Running times on small graphs (best viewed in color).} \label{fig:runtimesmall}
\vspace{-2mm}
\end{figure*}
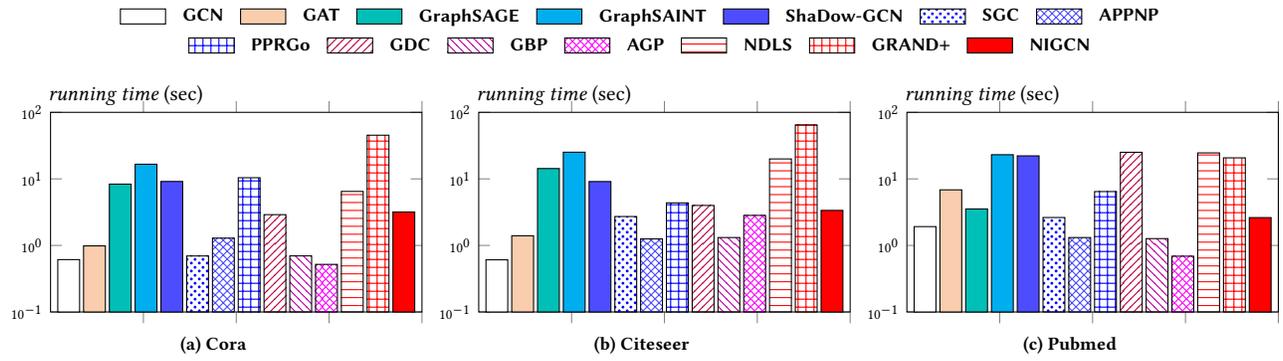

\section{Appendix}\label{sec:appendix}

\subsection{Proofs}\label{app:proof}
 
\begin{proof}[\bf Proof of Lemma~\ref{lem:theta}]
Before the proof, we first introduce {\it Chernoff Bound}~\cite{ChungL06} as follows.

\begin{lemma}[Chernoff Bound~\cite{ChungL06}]
Let $X_i$ be an independent variable such that for each $1\le i\le \theta$, $X_i \in[0,1]$. Let $X=\frac{1}{\theta}\sum^\theta_i X_i$. Given $\zeta \in (0,1)$, we have
\begin{equation}\label{eqn:hoeffding}
\Pr[{\mathbb E}[X] -X \ge \zeta]\le \e^{-\theta\zeta^2/2\mathbb{E}[X]}.  
\end{equation}
\end{lemma}
Let $v$ be an $\varepsilon$-importance neighbor of $u$ in the $\ell$-th hop. Suppose $v$ has been visited $\theta^{(\ell)}_v$ times after $\theta=\frac{2\eta^2}{\varepsilon}\log(\frac{W_s}{\delta \varepsilon})$ WRWs are generated. Let $\hat{\phi}(\ell, u, v)=\frac{U(\omega, \rho, \ell)\theta^{(\ell)}_v}{\Theta}$ be the weight estimation of $\phi(\ell, u, v)$. As Chernoff indicates, we have
\begin{align*}
& \Pr\Big[\phi(\ell, u, v)-\hat{\phi}(\ell, u, v) \ge \frac{\phi(\ell, u, v)}{\eta}\Big] \\
& \le \e^{-\frac{\theta \phi(\ell, u, v)^2}{\eta^2}/(2\phi(\ell, u, v))} \le \e^{-\frac{\theta\varepsilon}{2\eta^2}} \\
& \le \e^{-\frac{2\eta^2}{\varepsilon}\log(\frac{W_s}{\delta \varepsilon}) \cdot \frac{\varepsilon}{2\eta^2}}  = \frac{\delta \varepsilon}{W_s}. 
\end{align*}
Thus $\hat{\phi}(\ell, u, v) \ge \frac{\phi(\ell, u, v)}{\eta}$ holds with probability at least $1-\frac{\delta \varepsilon}{W_s}$. Within $L$ hops, the total weight is \[\sum^L_{\ell=0}w^{(\ell)}\sum_{v\in \V} (\D^{-1}\A)^\ell[u,v]=\sum^L_{\ell=0}w^{(\ell)}=W_s.\] Thus, the total number of $\varepsilon$-importance neighbors of $u$ is bounded by $\frac{W_s}{\varepsilon}$. By union bound, the total failure probability within $L$ hops is no more than $\frac{\delta \varepsilon}{W_s}\frac{W_s}{\varepsilon}=\delta$, which completes the proof.
\end{proof}

\subsection{Experimental Settings}\label{app:setting}

\spara{Dataset Statistics} Table~\ref{tbl:dataset} presents the detailed statistics of the $7$ tested datasets.

\spara{Running Environment} All experiments are conducted on a Linux machine with an NVIDIA RTX2080 GPU (10.76GB memory), Intel Xeon(R) CPU (2.60GHz), and 377GB RAM. 

\spara{Implementation Details} By following the state of the art~\cite{ChenWDL00W20, WangHW000W21, FengDHYCK022}, we implement and \gtro in PyTorch and C++. The implementations of \GCN and \APPNP are obtained from Pytorch Geometric\footnote{\url{https://github.com/pyg-team/pytorch\_geometric}}, and the other five baselines are obtained from their official releases. 

\spara{Parameter Settings} We mainly tune $\omega$, $\rho$, $\tau$, and $\varepsilon$ for \gtro. Table~\ref{tbl:parameter} reports the parameter settings adopted for each dataset. According to our analysis in Section~\ref{sec:diffmatrixlen}, we tune $\omega$ in $[0.9,1.5]$ and $\rho$ in $[0.01,0.1]$ for sparse datasets, \ie Cora, Citeseer, and Pubmed to capture long-range dependency by a smooth expansion tendency; we tune $\omega$ in $[1,10]$ and $\rho$ in $[0.5,1.5]$ to provide refined HKPR-alike properties, \ie reaching peak at certain hop. For extremely dense dataset Amazon, we tune $\omega$ in $[0.8,1.2]$ and $\rho$ in $[1,1.5]$ to realize optimized PPR-alike properties, \ie exponentially decrease in short distance. For $\tau$ and $\varepsilon$, we search $\tau$ in $[0.5,2]$ to find the best scaling factor for each dataset and tune $\varepsilon$ in the range of $[0.005,0.05]$. As stated, we fix $\eta=2$ and $\delta=0.01$.

\begin{table}[h]
 \caption{Dataset details.} \label{tbl:dataset}
    \setlength{\tabcolsep}{0.2em}
 	\begin{tabular} {@{}lrrrr@{}}
		\toprule
		{\bf Dataset}  & \multicolumn{1}{c}{\bf{\#Nodes ($\boldsymbol{n}$)}} & \multicolumn{1}{c}{\bf{\#Edges ($\boldsymbol{m}$)}} & \multicolumn{1}{c}{\bf \#Features ($\boldsymbol{f}$)} & \multicolumn{1}{c}{\bf \#Classes} \\ \midrule 
		{Cora}            & 2,708      &  5,429       & 	1,433     & 7     \\ 
		{Citeseer}        & 3,327      &  4,732       & 	3,703     & 6     \\ 
		{Pubmed}          & 19,717    & 44,338        & 500          & 3     \\ 
		{Ogbn-arxiv}      & 169,343    &  1,166,243   & 	128       & 40  	        \\ 
		{Reddit}          & 232,965    & 114,615,892  &   602       & 41     \\ 
		{Amazon}          & 1,569,960  &  132,169,734  &   200       & 107      \\ 
		{Papers100M}      & 111,059,956    &  1,615,685,872   & 	128       & 172     \\ \bottomrule 
	\end{tabular}
\end{table}

\begin{table}[!t]
 \caption{Hyper-parameters of \gtro.} \label{tbl:parameter}
 \centering
\begin{tabular} {@{}lllllll@{}}
\toprule
{\bf Dataset}  & \multicolumn{1}{c}{\bf $\omega$} & \multicolumn{1}{c}{\bf $\rho$} & \multicolumn{1}{c}{\bf $\tau$} & \multicolumn{1}{c}{\bf $\varepsilon$} & \multicolumn{1}{c}{\bf $\eta$} & \multicolumn{1}{c}{\bf $\delta$} \\ \midrule 
{Cora}            & 1.15     & 0.06   & 1.7   & 0.02    & 2  & 0.01 \\ 
{Citeseer}        & 1.1      & 0.04   & 1.2   & 0.03    & 2  & 0.01     \\ 
{Pubmed}          & 1.15     & 0.1    & 1.9   & 0.01    & 2  & 0.01 \\ 
{Ogbn-arxiv}      & 9.0      & 0.85	  & 1.5   & 0.008   & 2  & 0.01 \\ 
{Reddit}          & 4.0      & 1.1    & 0.8   & 0.008   & 2  & 0.01    \\ 
{Amazon}          & 0.9      & 1.15   & 1.5   & 0.05    & 2  & 0.01     \\ 
{Papers100M}      & 7.4      & 1.3    & 0.6   & 0.01    & 2  & 0.01  \\ \bottomrule 
\end{tabular}
\end{table}

\spara{Baselines for Inductive Learning} As stated, several baselines tested in transductive learning are not suitable for semi-supervised inductive learning. For example, \GraphSAINT samples a random subgraph as the training graph and demands each node in the subgraph owns a label~\cite{ZengZSKP20}. Nonetheless, the percentages of labeled nodes are small due to semi-supervised setting on large graphs (Ogbn-arxiv, Reddit, and Amazon), which degrades the performance of \GraphSAINT notably. \Grandp needs to sample a subset test nodes for loss calculation during training, which, however, is conflict with the setting of inductive learning.

\subsection{Additional Experiments}\label{app:exp}

\figurename~\ref{fig:runtimesmall} presents the running times of the $13$ tested methods in transductive learning. As shown, the efficiency of the tested methods on the three small datasets varies and there is no clear winners. In particular, \AGP (resp.\ \GCN) outperforms other methods on Cora and pubmed (resp.\ Citeseer). However, the F1-scores of \AGP and \GCN fall behind those of other models with a clear performance gap, as shown in Table~\ref{tbl:transductive}. Recall that \Grandp achieves the highest F1-scores on Cora and Citeseer, and our model \gtro performs best on Pubmed. Nonetheless, \Grandp runs up to $10\times \sim 20\times$ slower than \gtro does.